%% file: beck-fiala2.tex
\documentclass[11pt]{article}

\usepackage{graphicx}
\usepackage{latexsym}
\usepackage{amssymb}
\usepackage{amsmath}
\usepackage{amsthm}
\usepackage{bm}

\usepackage{mathtools}

\usepackage{nth}

\usepackage{forloop}
\usepackage[paper=letterpaper,margin=1in]{geometry}
\usepackage{paralist}
\usepackage{nicefrac}
\usepackage{rotating}
\usepackage{array,multirow}
\usepackage{threeparttable}
\usepackage{hyperref}

\newtheorem{thm}{Theorem}
\newtheorem{cor}[thm]{Corollary}
\newtheorem{conj}[thm]{Conjecture}
\newtheorem{lemma}[thm]{Lemma}
\newtheorem{defn}[thm]{Definition}

\newtheorem{obs}[thm]{Observation}

\newcommand\E{\mathbb{E}}
\newcommand\R{\mathbb{R}}

\newcommand\disc{\mathrm{disc}}
\newcommand\vecdisc{\mathrm{vecdisc}}

\begin{document}

\title{Improved Algorithmic Bounds for Discrepancy of Sparse Set Systems}

\author{Nikhil Bansal
\thanks{Department of Mathematics and Computer Science, Eindhoven University of Technology, Netherlands.  
Email:
\href{mailto:n.bansal@tue.nl}{n.bansal@tue.nl}.
Supported by a NWO Vidi grant 639.022.211 and an ERC consolidator grant 617951.}
\and
Shashwat Garg\thanks{Department of Mathematics and Computer Science, Eindhoven University of Technology, Netherlands.  
Email:
\href{mailto:s.garg@tue.nl}{s.garg@tue.nl}.
Supported by the Netherlands' Organisation for Scientific Research (NWO) under project no.~022.005.025.
}
}

\maketitle

\begin{abstract}
We consider the problem of finding a low discrepancy coloring for sparse set systems where each element lies in at most $t$ sets. We give an algorithm that finds a coloring with discrepancy 
$O((t \log n \log s)^{1/2})$ where $s$ is the maximum cardinality of a set. 
This improves upon the previous constructive bound of $O(t^{1/2} \log n)$ based on algorithmic variants of the partial coloring method, and for small $s$ (e.g.~$s=\textrm{poly}(t)$) comes close to  
the non-constructive  $O((t \log n)^{1/2})$ bound due to Banaszczyk. Previously, no algorithmic results better than $O(t^{1/2}\log n)$  were known even for $s = O(t^2)$.
Our method is quite robust and we give several refinements and extensions. 
For example, the coloring we obtain satisfies the stronger size-sensitive property that each set $S$ in the set system incurs an $O((t \log n \log |S|)^{1/2})$ discrepancy. 
Another variant can be used to essentially match Banaszczyk's bound for a wide class of instances even where $s$ is arbitrarily large. 
Finally, these results also extend directly to the more general Koml\'{o}s setting. 
\end{abstract}

\input{intro}

\input{prel}

\input{sec4}

\input{sec5}

\section*{Acknowledgements}
We would like to thank Thomas Rothvoss for a thorough reading and for several useful comments, leading to a cleaner presentation of the paper.

\bibliographystyle{alpha}
\bibliography{refr}

\end{document}

%% file: intro.tex
\section{Introduction}
Let $(V,\mathcal{S})$ be a finite set system, with $V=\{1,\ldots,n\}$ and $\mathcal{S} = \{S_1,\ldots,S_m\}$ a collection of subsets of $V$. For a two-coloring $\chi: V  \rightarrow \{-1,1\}$, the discrepancy of $\chi$ for a set $S$  is defined as $ \chi(S) =  |\sum_{j\in S} \chi(j) |$ and measures the imbalance from an even-split for $S$.
The discrepancy of the system $(V,\mathcal{S})$ is defined as 
\[ \disc(\mathcal{S}) = \min_{\chi:V \rightarrow \{-1,1\}} \max_{S \in \mathcal{S}} \chi(S). \]
That is, it is the minimum imbalance for all sets in $\mathcal{S}$, over all possible two-colorings $\chi$.

Discrepancy is a widely studied topic and has applications to many areas in mathematics and computer science. For more background we refer the reader to the books \cite{Chazelle,Mat99,Panorama}. In particular, discrepancy is closely related to the problem of rounding  fractional solutions of a linear system of equations to integral ones \cite{LSV,R12}, and is widely studied in approximation algorithms and optimization.

Until recently, most of the results in discrepancy were based on non-algorithmic approaches 
and hence were not directly useful for algorithmic applications. 
However, in the last few years there has been remarkable progress in our understanding of the algorithmic aspects of discrepancy \cite{B10,CNN11,LM12,Ro14,HSS14,ES14,NT15}. In particular, we can now match or even improve upon all known applications of the widely used partial-coloring method \cite{Spencer85,Mat99} in discrepancy.
This has led to several other new results in approximation algorithms \cite{R13,BCKL14,BN15,NTZ13}.


\paragraph{Sparse Set Systems:} 
Despite the algorithmic progress, one prominent question that has remained open is to match the known non-constructive bounds on discrepancy for low degree or sparse set systems. These systems are parametrized by $t$, that denotes the maximum number of sets that contain any element.
 Beck and Fiala \cite{BF81} proved, using an algorithmic iterated rounding approach, that any such set system has discrepancy at most $2t-1$. They also conjectured that the discrepancy in this case is $O(t^{1/2})$. Despite much effort, this remains an elusive open question.

The best known result in this direction is due to Banaszczyk \cite{B97}, who proved a discrepancy bound of $O((t \log n)^{1/2})$. Unlike most results in discrepancy that are based on the partial-coloring method,
Banaszczyk's proof is based on a very different and elegant convex geometric argument, and it is not known how to make it algorithmic.
Prior to Banaszczyk's result, the best known non-algorithmic bound\footnote{Throughout this paper we will assume that $t \geq \log^2 n$, to avoid taking minimum with the $O(t)$ bound every time.}  was $O(t^{1/2} \log n)$ \cite{Srin97}, based on the partial-coloring method. This bound was made algorithmic in \cite{B10}.
The question of matching Banaszczyk's bound algorithmically and several related variants of the problem have received a lot of attention in recent years \cite{N13,B13,ES14,EL15}.

The regime where the above questions are most interesting is perhaps when $t=\log^{\beta} n$ for some moderate value of $\beta$. 
If particular, if $\beta < 1$, then the $O(t)$ Beck-Fiala bound already beats the $O((t \log n)^{1/2})$ bound of Banaszczyk. On the other hand, if $\beta = 1/\epsilon$ for some small $\epsilon$, then the partial coloring bound already gives $O(t^{1/2} \log n) = O(t^{1/2+\epsilon})$, and becomes asymptotically close to the conjectured $O(t^{1/2})$ bound.

\subsection{Our Results}
In this paper we give new algorithmic results that improve upon the previously known algorithmic bounds for the Beck-Fiala problem in various ways.
Our first result is the following:
\begin{thm} 
\label{thm:main}
Given any $t$-sparse set system on $n$ elements, 
there is an efficient algorithm that with high probability, finds a coloring where each set $S \in \mathcal{S}$ incurs an  
$O((t \log n \log |S|)^{1/2})$ discrepancy.

In particular, this gives an $O((t \log n \log s)^{1/2})$ discrepancy coloring, where $s$ is the size of the largest set in the system.
\end{thm}

This bound is never worse than the previous $O(t^{1/2} \log n)$ algorithmic bound, and essentially matches Banaszczyk's bound if all sets have size $\textrm{poly}(t)$.
Previously, bounds better than  $O(t^{1/2} \log n)$ were known only for a very restricted range of $s$.
In particular, an $O((s \log t)^{1/2})$ (algorithmic) bound follows by a direct application of the Lov\'{a}sz Local Lemma (LLL). 
This gives an improvement over previous algorithmic bounds for $s \ll \min(t \log^2 n,t^2)$, but is much worse even for moderately larger values of $s$ such as $s=t^{10}$. It is unclear how to avoid the $s^{1/2}$ loss with LLL-based approaches\footnote{For a set of size $s$, a random coloring incurs $\Omega(s^{1/2})$ discrepancy with probability almost $1$.}, and the problem of improving upon the $O(t^{1/2}\log n)$ bound was open even for the case of $s=t^2$.

As a concrete example to summarize the status of the known bounds: for $t = O(\log^2 n)$ and $s = t^2=O(\log^4 n)$, the conjectured Beck-Fiala bound is $O(\log n)$,  the best known non-constructive bound is $O(\log^{3/2} n)$, all previous algorithms incur $\Omega(\log^2 n)$ discrepancy, and Theorem \ref{thm:main} above  achieves $\tilde{O}(\log^{3/2}n)$ where $\tilde{O}(\cdot)$ hides $\textrm{poly}(\log \log n)$ factors.


The algorithm is based on a random walk based approach, similar to that in \cite{B10}, where each step of the walk is guided by a solution to a semidefinite program (SDP). However, as we describe in Section \ref{s:overview}, the previous partial coloring based approaches seem inherently incapable of improving upon the $O(t^{1/2} \log n)$ bound (in fact this is true even for $s=t$, where LLL based techniques perform much better). Moreover, the previous approaches also cannot give a size-sensitive discrepancy bound for each set $S$, as in Theorem \ref{thm:main}. 
A key new idea of our algorithm is to add some extra constraints to the SDP which ensure that {\em every} large enough set makes enough ``progress" towards getting colored, and to use a result of Nikolov \cite{N13} to argue the feasibility of this SDP. 
Roughly, these constraints ensure that a set of size $|S|$ incurs an $O(\sqrt{t})$ discrepancy in expectation for only $O(\log |S|)$ phases. 
We give an overview of the algorithm in Section \ref{s:alg} after defining the relevant background in Section \ref{s:prel}.


While Theorem \ref{thm:main} does not give anything new directly if $s$ is arbitrarily large, our method is quite versatile and readily lends to many extensions. 
These can be used to match or almost match Banaszczyk's bound in many cases even when $s$ is arbitrary.
These are discussed in Theorems \ref{thm:ext1}, \ref{thm:ext2} and their corollaries  below.


\begin{thm}
\label{thm:ext1}
For $i=1,\ldots,\log n$, let $m_i$ denote the number of sets in $\mathcal{S}$ with size in $(2^{i-1},2^{i}]$. Call these class-$i$ sets.
Then, there is an algorithm that finds a coloring with discrepancy $O((t i (\log m_i+\log\log n))^{1/2})$ for each class-$i$ set.
\end{thm}

The bound above  refines Theorem \ref{thm:main} (which follows as $m_i \leq n$ for each $i$) and can lead to substantially improved bounds when the number of large  sets is not too high.
Recall that in general we can assume that $m_i$ is not too large, as by a standard linear algebraic argument \cite{BF81}, one can assume that $n \leq m$ which implies that 
the average row size is also at most $t$ and that $m_i \leq nt/2^{i-1}$ for every $i$.
The following corollary of Theorem \ref{thm:ext1} shows that if the system satisfies
the stronger bound of $m_i \leq n^{(\log^{O(1)}t)/i}$ instead, then one can essentially match Banaszczyk's bound.

\begin{cor}
\label{cor1}
If the number of class-$i$ sets is at most $ n^{(\log^{O(1)}t)/i}$, then there is an algorithm that finds an $\tilde{O}((t \log n)^{1/2})$ discrepancy coloring, where $\tilde{O}(\cdot)$ hides $\mbox{poly}(\log {t}\log\log {n})$ factors.
\end{cor}

%
%

Theorem \ref{thm:main} also extends to the more general setting of the Koml\'{o}s conjecture (defined in Section \ref{s:prel}) where the matrix entries can be arbitrary reals instead of only $0$-$1$. For this setting, Banaszczyk's result \cite{B97} implies an $O(\log^{1/2} n)$ non-constructive bound.
\begin{thm}
\label{thm:kom}
Given a matrix $A$ with $n$ columns such that the $\ell_2$-norm of any column is at most $1$, then there is an algorithm to find a coloring where each row $j$ incurs a discrepancy of
$O( (\log n \log s_j)^{1/2})$, where $s_j$ is the $\ell_1$-norm of row $j$.
\end{thm}

The following extension of Theorem \ref{thm:kom} allows further flexibility in choosing different weights $w_j$ for rows to allow improved bounds in many cases. 
\begin{thm}
\label{thm:ext2}
Suppose there is a reweighting of the rows $j$ with non-negative weights $w_j$ such that the $\ell_2$-norm of any column is at most $\beta$, then 
the  algorithm gets a discrepancy of  $O(\beta(\log n\log (w_j s_j))^{1/2}/w_j)$   for each row  $j$.
\end{thm}

The following corollary of Theorem \ref{thm:ext2} shows that Banaszczyk's bound can be almost matched algorithmically if the sparsity of set system decreases suitably when restricted to very large sets, i.e.~no element lies essentially in only very large sets.

\begin{cor}
\label{cor2}
Given a $t$-sparse set system with the property that each element lies in $O(t/i^{1+\epsilon})$ 
sets of size more than $t 2^{i-1}$ for each $i\geq 1$, for some fixed $\epsilon>0$, then there is an algorithm to find an $O((t \log n \log t)^{1/2})$ discrepancy coloring.
\end{cor}

%
%



%% file: prel.tex
\section{Preliminaries}
\label{s:prel}
We describe some of the basic concepts that we will need.

\subsection{Vector Discrepancy}
If $A$ is the $m\times n$ incidence matrix of the set-system $(V,\mathcal{S})$ with rows corresponding to the sets $S_1,\ldots,S_m$ and columns to the elements $[n]\footnote{We use $[n]$ to denote the set \{1,2,\ldots,n\}.}$, then we can write the discrepancy of $A$ as 
\[ \disc(A)= \min_{x \in \{-1,1\}^n} \|Ax\|_\infty \]
The above definition of discrepancy generalizes to any matrix $A$ with real entries.
Let us index the rows of $A$ by $j$ and the columns by $i$.
If we relax the entries $x_i$ to be unit vectors ${v}_i$ in $\R^n$ instead of $\pm 1$, we obtain the relaxation
\[ \vecdisc(A) = \min_{{v}_1,\ldots,{v}_n \in S^{n-1}} \max_{j=1}^m \|\sum_{i=1}^n A_{ji} {v}_i \|_2. \]
We refer to this as the vector discrepancy of $A$. The vector discrepancy can be computed efficiently (to arbitrarily high accuracy) by solving the following semidefinite program (SDP) and doing a binary search on $\lambda$.  
\begin{eqnarray*}
  \|\sum_{i }A_{ji} v_i \|_2^2 & \leq &  \lambda^2 \qquad  \forall j\in [m]\\
   \|v_i \|_2^2 & = &  1  \qquad \forall i \in [n]
  \end{eqnarray*}

\paragraph{Koml\'{o}s Conjecture:}
The following conjecture generalizes the Beck-Fiala conjecture. 
\begin{conj}[Koml\'{o}s] Any matrix $A$ with columns of $\ell_2$-norm at most $1$ has discrepancy $O(1)$. 
\end{conj}
The Beck-Fiala conjecture follows by scaling the columns of the incidence matrix $A$ by $\sqrt{t}$. While the best known bound for the Koml\'{o}s setting 
is $O(\log^{1/2} n)$ \cite{B97},
Nikolov \cite{N13} showed that the Koml\'{o}s conjecture holds for vector colorings. In particular, he showed the following result:
\begin{thm}[\cite{N13}]
\label{thm:sasho}
For any $m\times n$ matrix $A$ with columns of $\ell_2$-norm at most $1$, $\vecdisc(A) \leq 1$.
\end{thm}
As pointed out to us by Raghu Meka, an easy proof of $\vecdisc(A) \leq O(1)$ also follows
from Banaszczyk's result \cite{B97}.

\subsection {Martingales}
Let $X_1,X_2,\ldots,X_n$ be a sequence of independent random variables on some probability space, and let 
$Y_t$ be a function of $X_1,\ldots,X_t$. The sequence 
$Y_0,Y_1,Y_2,\ldots,Y_n$ is called a martingale with respect to the sequence $X_1,\ldots,X_n$ if for all $t \in [n]$, 
$\E[|Y_t|]$ is finite and $\E[Y_t| X_1,X_2,...,X_{t-1}]= Y_{t-1}$. We will use $\E_{t-1}[Y_t]$ to denote $\E[Y_t | X_1,X_2,...,X_{t-1}]$.
We will need the following martingale concentration inequality.

%
%

\begin{thm}[Freedman \cite{Freedman}]
\label{thm:freedman}
Let $Y_1,\ldots,Y_n$ be a martingale with respect to $X_1,\ldots,X_n$ such that 
$|Y_t - Y_{t-1}| \leq M$ for all $t$, and let $W_t = \sum_{j=1}^t \E_{j-1}[(Y_j - Y_{j-1})^2] = \sum_{j=1}^t
\textrm{Var}[Y_j|X_{1},\ldots,X_{j-1}]$.
Then for all $\lambda \geq 0$  and $\sigma^2 \geq 0$, we have 
\[ \Pr[|Y_n - Y_0| \geq \lambda \textrm{ and } W_n \leq \sigma^2]\le 2 \exp\left(-\frac{\lambda^2}{2 (\sigma^2 + M \lambda/3)} \right). \] 
\end{thm}

%% file: sec4.tex
\section{The Main Result}
\label{s:alg}
In this section we describe the main algorithm and prove Theorem \ref{thm:main}. 

\subsection{High-level Overview}
\label{s:overview}
 The algorithm has a similar structure to previous random walk based approaches.
The algorithm starts with the coloring $x_0=0^n$ at time $0$, and at each time step $k$, updates the coloring at time $k-1$ by adding a small increment. This increment is determined by solving an appropriate SDP.
If a variable reaches $-1$ or $1$ it is frozen, and its value is not updated any more. The variables that are not frozen are called alive.

\paragraph{Limitations of partial coloring based approaches:}
However, the previous algorithms based on these approaches are only able to give a discrepancy of $O(t^{1/2} \log n)$ for the following reason. Roughly speaking, the execution of the algorithm can be divided into $O(\log n)$ phases. In each phase, about half the variables get frozen, and the updates of the coloring at each time step are chosen so that each set incurs an expected discrepancy of $O(\sqrt{t})$ during a phase.
So after $O(\log n)$ phases each set incurs an expected discrepancy of $O((t \log n)^{1/2})$, and to bound the maximum discrepancy, one then takes a union bound over the $m$ sets and loses an additional $\log^{1/2} m\approx \log^{1/2} n$ factor.

If one could show a stronger statement that in each phase, the number of alive variables for {\em every} set
reduces by a constant factor, then each set $S_j$ would get colored in at most $O(\log |S_j|)$ phases, and the expected discrepancy for a set would be $O((t \log |S_j|)^{1/2}$  (instead of the $O((t \log n)^{1/2})$ previously). Theorem \ref{thm:main} would then follow after losing another $\log^{1/2} n$ factor for the union bound over the sets. 

However, the property that every set reduces in size is non-trivial to guarantee. For example, even in the non-constructive partial coloring based methods \cite{Mat99}, in each partial coloring step (i.e.~a phase) one can only guarantee that half the remaining variables (globally) get colored, and there could be several sets that incur $\Omega(\sqrt{t})$ discrepancy even though very few of their elements get colored. In particular, these methods incur an $O(t^{1/2} \log n)$ discrepancy even if all the sets are of size $t$.

\paragraph{New idea:}
The main idea of our algorithm is to add additional SDP constraints, given by (\ref{sdp2}), that we call the energy constraints that ensure that the ``energy" of {\em each} (large enough) set increases significantly in each phase. 
This can be used to show that for each set $S\in \mathcal{S}$, essentially all of its elements will get colored in $O(\log |S|)$ phases. 
In contrast, the previous algorithms only ensured a global energy constraint over all the variables to measure progress of the algorithm. Moreover, the feasibility
problem for this SDP can be viewed as a vector discrepancy problem for another matrix where the $\ell_2$-norm remains $O(1)$. Theorem \ref{thm:sasho} then guarantees feasibility.
The main technical part of the analysis is to show that these new energy constraints are sufficient to guarantee the progress property for each set.

Theorems \ref{thm:main} and \ref{thm:kom} follows quite directly from the above approach. The extensions in theorems \ref{thm:ext1} and \ref{thm:ext2} also follow the general approach above, but use some additional flexibility in the approach to trade off and balance the different parameters.

\subsection{Algorithm}

We will index time by $k$.
Let $x_k\in [-1,1]^n$ denote the coloring at the end of time step $k$. During the algorithm,
variables which get set to at least $(1-1/n)$ in absolute value are called frozen and their values are not changed any more. The remaining variables are called alive. We denote by $A(k)$ the set of alive variables at the end of time step $k$. Initially all variables are alive. Let $\gamma = n^{-2}$, and let $T= (10/\gamma^2) \log n$.



\begin{enumerate}
\item Initialize $x_0(i) =0$ for all $i\in [n]$ and $A(0)=\{1,2,...,n\}$. 
\item  For each time step $k=1,2,\ldots, T$ repeat the following:
\begin{enumerate}
\item \label{apx:step1}
 Find a feasible solution to the following semidefinite program:
\begin{eqnarray}
\label{sdp1}  \|\sum_{i \in S_j} v_i\|_2^2 & \leq &  \ 2t \qquad  \textrm{for each set }  S_j\\
\label{sdp2}  \| \sum_{i\in S_j}x_{k-1}(i) v_i\|_2^2  & \leq & \ 2t \qquad \textrm{for each set } S_j\\
   \|v_i\|_2^2 & = &  1  \qquad \forall i \in A(k-1) \nonumber \\
   \| v_i\|_2^2 & = & 0  \qquad \forall i \notin A(k-1) \nonumber
  \end{eqnarray}
\item 
\label{apx:round}
Construct $\gamma_k \in \mathbb{R}^n$ as follows: 
let $r \in \R^n$ be a random $\pm 1$ vector, obtained by setting each coordinate $r(i)$ independently to $-1$ or $1$ with probability $1/2$.

For each $i \in [n]$, let $\gamma_k(i) = \gamma \langle r,v_i \rangle.$
Update $x_k = x_{k-1} + \gamma_k$. 

\item
\label{apx:rnd}
 Initialize $A(k)=A(k-1)$.

For each $i$, if $x_k(i)\geq 1-1/n$ or  $x_k(i) \leq -1+1/n$, update $A(k) = A(k) \setminus \{i\}$.
\end{enumerate}
\item 
\label{stp3}
Generate the final coloring as follows.
Set $x_T(i)=1$ if $x_T(i)\geq 1-1/n$ and set $x_T(i)=-1$ if $x_T(i) \leq -1+1/n$.
If $i \in A(T)$, set $x_T(i)$ arbitrarily to $\pm 1$. 
\end{enumerate}

\subsection{Analysis}

We begin with some simple observations.

\begin{lemma}
 The SDP is feasible at each time step $k$.
\end{lemma} 
\begin{proof}
Consider the incidence matrix $A$ of the set system. 
For each row $A_j$ of  $A$ corresponding to $S_j$, we add another row 
$A'_j = A_j\textrm{ diag}(x_{k-1})$, where $\textrm{diag}(x_{k-1}) $ is the $n\times n$ diagonal matrix with the $(i,i)$ entry
$x_{k-1}(i)$. In other words, $A'$ is obtained by scaling the $i$-th column of $A$ by $x_{k-1}(i)$.
As $|x_{k-1}(i)|\le 1$ for each $i$, the $\ell_2$-norm of this augmented matrix $A\cup A'$ is at most $\sqrt{2t}$.

As the SDP in the algorithm above corresponds to finding a vector coloring (restricted to the alive elements) 
with discrepancy $\sqrt{2t}$ for the $2m \times n$ augmented matrix $A \cup A'$, Theorem \ref{thm:sasho} implies that it has a feasible solution. 
\end{proof}

%

\begin{lemma}
\label{lem:tri}
For any vector $v \in \R^n$ and a random $\pm 1$ vector $r \in \R^n$,  $\E[\langle r,v\rangle^2] = \|v\|_2^2$   
and $|\langle r,v\rangle | \leq \sqrt{n}\|v\|_2$.
\end{lemma}
\begin{proof}
Writing $v$ in terms of its coordinates $v=(v(1),\ldots,v(n))$,  
\[ \E[\langle r,v\rangle^2] =  \E [(\sum_i  r(i) v(i) )^2]  = \sum_{i,j} \E[r(i)r(j)] v(i)v(j)  = \|v\|_2^2 \]
where the last equality uses that $\E[r(i)r(j)]=0$ for $i \neq j$  and $\E[r(i)^2]=1$.

The second part follows by Cauchy-Schwarz inequality, as $|\langle r,v\rangle | \leq \|r\|_2 \|v\|_2 = \sqrt{n}\|v\|_2$.
\end{proof}


This implies the following.

\begin{obs}
As $\|v_i\| \leq 1$, $|\gamma \langle r,v_i \rangle |\leq \gamma\sqrt{n} \|v_i\|_2 \leq 1/n$, the rounding in step \ref{stp3} affects the discrepancy of any set by at most $n \cdot (1/n)=1$. So we can ignore this rounding error.
Moreover, $|\gamma \langle r,v_i \rangle | \leq 1/n$ also implies that no $x_k(i)$ goes out of the range $[-1,1]$ during any step of the algorithm.
\end{obs}

Let us divide the execution of the algorithm into $ p=10\log n$ phases each consisting of $L:=1/\gamma^2$ time steps.
The following lemma says that the discrepancy added in a few phases cannot be too large.

\begin{lemma}
\label{lem:initial-disc}
Fix a set $S\in\mathcal{S}$ and some $\ell \in [p]$. 
Then, for any $ 0 \leq \beta  \leq  1/(\gamma\sqrt{n}) = n^{3/2}$, the discrepancy of $S$, $disc(S)$, after $\ell$ phases satisfies
\[ \Pr[\disc(S) \geq  \beta \sqrt{2 t \ell} ] \leq 2 \exp(- \beta^2/3) \]
\end{lemma}
\begin{proof}
Let $Y_k = \disc_{k}(S)$ denote the discrepancy of $S$ after time $k$. Then, $Y_k = Y_{k-1} + \gamma \langle r^k, \sum_{i \in S} v^k_i\rangle$, where $r^k$ denotes the random $\pm 1$ vector at time $k$ and $v_i^k$ is the vector for $i$ in the SDP solution at time $k$.
As $\E[r^k,v]=0$ for any vector $v$ and the $r^k$ are independent, $Y_k$ forms a martingale with respect to $r^1,r^2\ldots$.

So we can apply Theorem \ref{thm:freedman}. By the SDP constraint \eqref{sdp1} we can bound,
\[ M\le |\gamma \langle r^k, \sum_{i \in S} v^k_i\rangle| \leq \gamma \|r^k\|_2 \|\sum_{i \in S} v^k_i\|_2 \leq  \gamma n^{1/2} \cdot \sqrt{2t} \]

Similarly, by Lemma \ref{lem:tri} ,
\[\E_{k-1}[(Y_k-Y_{k-1})^2] = \E[\gamma^2  \langle r^k, \sum_{i \in S} v^k_i\rangle^2] = \gamma^2 \|\sum_{i \in S} v^k_i \|_2^2 \leq 2 t \gamma^2. \]  
This implies that $W_k \leq 2 kt \gamma^2$ for all $k$.  Applying Freedman's inequality from Theorem \ref{thm:freedman} with these bounds on $M$ and $W_k$, we obtain

\begin{eqnarray*}
 \Pr[\disc_{\ell L}(S) \geq   \beta \sqrt{2 t \ell} ] & \leq & 2 \exp \left( -\frac{\beta^2 2t \ell }{2 (W_{\ell L} + M \beta \sqrt{2t \ell}/3) } \right) \\
& \leq  & 
 2 \exp\left(- \frac{2 t \ell \beta^2}{2 (2\ell L t \gamma^2 + \gamma \beta \sqrt{2nt} \sqrt{2t \ell} /3)} \right) \\
& \leq &  2 \exp\left(- \frac{\beta^2}{(2 + \gamma \beta \sqrt{n})} \right)  \leq 2 \exp(-\beta^2/3).
\end{eqnarray*}

The last step uses that $\beta \gamma \sqrt{n} \leq  1$. 
\end{proof}

For a set $S$, let $A_S(k)$ denote the number of alive elements in set $S$ at the end of time step $k$.
\begin{defn}[Energy Deficit]
For a set $S\in\mathcal{S}$ and given a coloring $x_k$ at the end of time $k$, the {\em energy deficit} of $S$ is defined as 
 $D_S(k) = \sum_{i \in S} (1-x_k^2(i))$.
\end{defn}

Initially at time $k=0$, the energy deficit of any set $S\in\mathcal{S}$ is $|S|$, and is always non-negative. We also note the following.

\begin{lemma} 
\label{l:dsas}
For any set $S\in\mathcal{S}$ and time $k$, it holds that the number of alive elements $ A_S(k) \geq D_S(k) -2 $.
\end{lemma}
\begin{proof}
Any frozen element can contribute at most  $1- (1-1/n)^2 \leq 2/n$ to the energy deficit. Moreover, since each alive alive element can contribute at most $1$, $D_S(k) \leq A_S(k) + n \cdot (2/n)$ and the result follows.  
\end{proof}

We now come to the  key part of the proof, which shows that for each large enough set, its energy deficit decreases substantially in each phase, with sufficiently high probability. 


\begin{thm}
\label{thm:winwin}
Fix a set $S\in\mathcal{S}$ and a time $k$. For any $\lambda \geq 0$, 
with probability at least $1-4\exp(-\lambda^2/100t)$, it holds that
\begin{equation}
\label{key:eq} 
D_S(k+L) \leq  D_S(k) -  A_S(k+L) + \lambda.
\end{equation}
\end{thm}

As shown in the corollary below, this theorem implies that a set $S$ shrinks to size $O(\lambda)$ in $O(\log |S|)$ phases with high probability. While these results are stated for general $\lambda$ (which is useful for the extensions later), the reader may think of $\lambda = O((t \log n)^{1/2})$.



\begin{cor}
\label{cor:sixteen}
Fix a set $S\in\mathcal{S}$ and $\lambda \ge 1$. Let $\ell = 1+ 3\log (|S|/\lambda)$. Then, 

\[ \Pr \left[ A_S (\ell L)   \geq 2\lambda + 2 \right] \leq 16 \log n \exp\left(-\frac{\lambda^2}{100t} \right). \]
That is, the probability that $S$ has more than $2\lambda + 2$ alive elements after $\ell$ phases is at most 
$16\log n \exp(-\lambda^2/100t)$.
\end{cor}
\begin{proof}
Let us apply Theorem \ref{thm:winwin} at $k=0,L, \ldots, (\ell-1)L$. Then, with probability at least $1-4\ell \exp(-\lambda^2/100t) \geq 1- 16 \log n \exp(-\lambda^2/100t)$, we can assume that \eqref{key:eq}  holds at each of these steps. Let us condition on this event.
If $A_S(\ell L) \geq 2\lambda+2$, then for each of the time steps $k=0,L,\ldots,(\ell-2)L$, 
\begin{eqnarray*} D_S(k+L) & \le &  D_S(k)- A_S(k+L) + \lambda \\
                           & \le & D_S(k) - A_S(k+L)/2 - 1  \qquad (\textrm{as } \lambda \leq A_S(\ell L)/2 -1 \leq A_S(k+L)/2 -1)\\ 
                           & \leq &  D_S(k) - D_S(k+L)/2  \qquad \qquad (\textrm{By Lemma }\ref{l:dsas})
\end{eqnarray*}
This implies that $D_S((\ell-1) L ) \leq D_S(0) (3/2)^{-\ell+1} \leq \lambda$.
However, applying \eqref{key:eq} again at $k=(\ell-1)L$ gives that 
\[D_S(\ell L) \leq D_S((\ell-1)L -A_S(\ell L) + \lambda  \leq \lambda -  (2 \lambda +2) + \lambda < 0\]
 which contradicts that the energy deficit is always non-negative.   
%
\end{proof}

Before we prove theorem \ref{thm:winwin}, let us first see how this implies Theorem \ref{thm:main}. 

\bigskip

\noindent{\em Proof  (Theorem \ref{thm:main}).}
Applying corollary \ref{cor:sixteen} with $\lambda = 20 (t \log n)^{1/2}$, it follows that a set $S_j$ reduces to size 
at most $O((t \log n)^{1/2})$ in at most $1+3(\log |S_j|)$ phases with probability at least $1- 16\log n \exp(-\lambda^2/100t) \geq 1-1/n^3$.
By a union bound, this holds for all the (at most $nt$) sets with probability at least $1-1/n$.

By Lemma \ref{lem:initial-disc} with $\beta = 4 \log^{1/2} n$ and  $\ell_j=1+3 \log |S_j|$, the discrepancy $S_j$ during the first $\ell_j$ phases is at most $ O(\beta (t\ell_j)^{1/2}) = O((t \log n \log |S_j|)^{1/2})$ with probability $1-O(n^{-5})$. 
When $S_j$ has fewer than $\lambda = O((t \log n)^{1/2})$ alive elements, it is {\em safe} and can incur at most $\lambda$ additional discrepancy during the rest of the algorithm. Taking a union bound over all the sets gives the result.
{\hfill$\qed$\bigskip}

We now prove Theorem \ref{thm:winwin}.

\bigskip

\noindent{\em Proof (Theorem \ref{thm:winwin}.)}
Let $v_i^h$ denote the vector $v_i$ in the SDP solution at time step $h$,  and let $r^h$ denote the random vector at time $h$.
Let $\Delta_S(h):= D_S(h)-D_S(h-1)$ denote the change in energy deficit of $S$ at time $h$. Then,
\begin{eqnarray}
\Delta_S(h) & = &   \sum_{i \in S} (x_{h-1}(i)^2 - x_{h}(i)^2 ) =
              \sum_{i \in S} (x_{h-1}(i)^2 - (x_{h-1}(i) + \gamma \langle r^h,v_i^h \rangle )^2 ) \nonumber \\
						& = &   - 2\gamma \langle r^h, \sum_{i\in S}x_{h-1}(i) v_i^h \rangle - \sum_{i\in S} \gamma^2 \langle r^h,v_i^h\rangle^2  
						\label{eq:energy}
\end{eqnarray}


Let us denote the first and second terms on the right hand side of \eqref{eq:energy} by \[T_h:=  - 2\gamma \langle r^h, \sum_{i\in S}x_{h-1}(i) v_i^h \rangle  \quad \textrm{ and } \quad  R_h:=  - \sum_{i\in S} \gamma^2 \langle r^h,v_i^h\rangle^2  \] 
We need to show that with probability at least $1-4\exp(-\lambda^2/100 t)$,
\[ D_S(k+L) - D_S(k) =  \sum_{h=k+1}^{k+L}\Delta_S(h) =\sum_{h=k+1}^{k+L} (T_h + R_h)\leq - A_S(k+L) + \lambda.\]

To do this, we will show that $T_h$ is a martingale with small variance (this is where the SDP constraints \eqref{sdp2} will be used), and that $\sum_{h=k+1}^{k+L} R_h$ is almost deterministic (its variance can be made arbitrarily small by making $\gamma$ small) with $\E[ \sum_{h=k+1}^{k+L} R_h] \leq - A_S(k+L)$.


Let $Y_h = \sum_{h'=1}^h T_{k +h'}$ and $X_h = r^{k+h}$. Then $Y_0=0$ and  $Y_1,\ldots,Y_L$ form a martingale with respect to $X_1,\ldots,X_L$. 
This is the same martingale as in the proof of Lemma \ref{lem:initial-disc}, except that it is multiplied by a factor of $-2$.
So applying Freedman's inequality as in the proof of Lemma \ref{lem:initial-disc} with $\ell=1$ gives  
\[ \Pr[|Y_L| \geq 2\beta \sqrt{2t }] \leq 2 \exp(-\beta^2/3).\]
Choosing $\beta = \lambda/(4\sqrt{2t})$ gives
\begin{equation}
\label{eq:29} 
 \Pr[ |Y_L| \geq \lambda/2 ] \leq 2 \exp(-\lambda^2/96t). 
\end{equation}

We now bound $\sum_h R_h$. 
Let $Z'_h = \sum_{h'=1}^h R_{k +h'}$, and define the martingale 
\begin{equation}
\label{eq:doob}
 Z_h = Z'_h - \sum_{h'=1}^h \E_{h'-1}[Z'_{h'}-Z'_{h'-1}]  = Z'_h - \sum_{h'=1}^h \E_{h'-1}[R_{k+h'}]
 \end{equation}
 with respect to $X_1,\ldots,X_h$ (this is the standard Doob decomposition of $Z'_h$). 
 
 We now apply Freedman's inequality to $Z_h$. By \eqref{eq:doob}
\[|Z_h - Z_{h-1}|  =  |R_{k+h} - \E_{h-1}[R_{k+h}]| \leq 2|R_{k+h}| \leq 2\gamma^2n|S| \leq 2\gamma^2 n^2.\] and thus $M \leq 2\gamma^2 n^2$.    
Moreover using the trivial bound  $\E_{h-1}[(Z_h - Z_{h-1})^2 ] \leq M^2 = 4\gamma^4 n^4$, we obtain that total variance $W_L \leq L M^2 = 4\gamma^2 n^4 = 4$. By Freedman's inequality with $\lambda \leq n$, we get
\begin{equation}
\label{eq:30}
 \Pr [ Z_L > \lambda/2 ]  \leq  2 \exp( -\lambda^2 / 8(W_L + M\lambda/6))  \leq 2\exp(- \lambda^2/40) \leq 2\exp(- \lambda^2/40t).
\end{equation}

 By Lemma \ref{lem:tri} for each $h \in [k+1,k+L]$,  
 \[ \E[R_h] =  -\gamma^2\sum_{i\in S}  \|v_i^h\|^2 =  -\gamma^2A_S(h-1) \leq  -\gamma^2A_S(k+L).  \] 
So we can bound $\sum_{h'=1}^h R_{k +h'}$ as
\begin{eqnarray}
\sum_{h'=1}^h R_{k +h'} & = &  Z'_L = Z_L + \sum_{h'=1}^L \E_{h'-1} [R_{k+h'}] \nonumber \\
 & \leq &   Z_L  - L \gamma^2 A_S(k+L) =  Z_L  -A_S(k+L). \label{eq:31}
\end{eqnarray}
This gives that
\begin{eqnarray*}
\Pr[ \sum_{h=k+1}^{k+L} (T_h+R_h) \geq - A_S(k+L) + \lambda ] & = &   \Pr[ Y_L+Z_L^\prime\geq - A_S(k+L) + \lambda] \\
& \leq  &\Pr[ Y_L+Z_L \geq \lambda]   \qquad \textrm{(by \eqref{eq:31})}\\                                      
														& \leq &  \Pr[ |Y_L| \geq \lambda/2]  + \Pr[ Z_L \geq \lambda/2] \\  
														& \leq &  4 \exp(-\lambda^2/100t) \qquad \textrm{(by \eqref{eq:29} and \eqref{eq:30}) }
\end{eqnarray*}

{\hfill$\qed$}

%% file: sec5.tex
\section{Extensions}
\label{s:ext}

We now describe the various extensions. The first extension is to the Koml\'{o}s setting.

\bigskip

\noindent{\bf Theorem \ref{thm:kom} (restated):} {\em
Given a matrix $A$ with $n$ columns such that the $\ell_2$-norm of any column is at most $1$, then there is an algorithm to find a coloring where each row $j$ incurs a discrepancy of
$O( (\log n \log s_j)^{1/2})$, where $s_j$ is the $\ell_1$-norm of row $j$.
}

\medskip

The proof follows along the same lines as before, and only needs a modification in the definition of the energy deficit. We describe the relevant differences.

\begin{proof}(Sketch)
For the Koml\'{o}s setting, we replace the SDP constraints \eqref{sdp1} by $\|\sum_{i} a_{ji} v_i \|^2 \leq 2t $ and the constraints \eqref{sdp2} by  $\| \sum_{i} |a_{ji}| x_{k-1}(i) v_i \|^2 \leq 2t$  (note the absolute value in the second set of constraints) where $a_{ji}$ denotes the $(j,i)^{th}$ element of $A$. 
To use the latter constraint,
we define the energy deficit for row $j$ at time $k$ as  $D_j(k) =s_j  - \sum_{i} |a_{ji}| x_{k}(i)^2$. 
The energy deficit is always non-negative and is $s_j$ at time $k=0$. 
We also define $A_j(k)$ now as the $\ell_1$-norm of row $j$ restricted to the alive variables. This gives that,
\[D_j(k)-D_j(k-1) = -2 \gamma \sum_i |a_{ji}| \langle r^k,x_{k-1}(i) v^k_i \rangle  - \gamma^2 \sum_i |a_{ji}| \langle r^k,v^k_i\rangle^2.\]
Exactly as in the proof of Theorem \ref{thm:winwin}, the energy constraints of the SDP ensure that the first term behaves as a martingale with deviation in each step  $O(\gamma t^{1/2})$. Similarly the second term causes the deficit to decrease in expectation by at least 
$\gamma^2 \sum_{i \in A_j(k)} |a_{ji}| \langle r^k,v^k_i\rangle^2$, which exactly as before, gives us the required decrease when the $\ell_1$-norm of the alive variables in higher than our threshold $O(\lambda)$.
\end{proof}

The next extension follows by using the flexibility in choosing the $\lambda$'s differently for different sets.

\bigskip 

\noindent{\bf Theorem \ref{thm:ext1} (restated):} {\em
For $i=1,\ldots,\log n$, let $m_i$ denote the number of sets in $\mathcal{S}$ with size in $(2^{i-1},2^{i}]$. Call these class-$i$ sets.
Then there is an algorithm that with probability at least $1-1/\log n$ finds a coloring with discrepancy $O((ti ( \log m_i+\log\log n))^{1/2})$ for each class-$i$ set.
}

\medskip

\begin{proof}
For class-$i$ sets $S$, we set $\lambda_i = 20 (t (\log m_i + \log \log n))^{1/2}$.
Then by Corollary \ref{cor:sixteen}, the probability that a class-$i$  set does not become smaller than $2 \lambda_i$ in $3\log 2^i =3i$ phases is at most $16/(m_i^4 \log^3 n)$. 
By a union bound over all the sets in class $i$, and over all the $\log n$ choices of $i$, each set decreases to a size at most $2\lambda_i$ after $O(i)$ phases with probability at least $1-1/(2\log n)$.
 
Using lemma \ref{lem:initial-disc} with $\beta=3 (\log m_i+\log\log n)^{1/2}$, probability that a class-$i$ set $S$ gets discrepancy more than $3 (6 \cdot ti(\log m_i + \log\log n))^{1/2}$ in the first $3i$ phases is at most $2/(m_i^3\log^3 n)$. Taking a union bound over all the $m_i$ sets in class $i$ and over all the  $\log n$ choices of $i$, we get that each set gets discrepancy $O((ti(\log m_i+\log\log n))^{1/2})$ in the first $3i$ phases with probability at least $1-1/(2\log n)$.
 
Together, this implies that with probability at least $1-1/\log n$, the discrepancy for every class-$i$ set is $O((ti ( \log m_i+\log\log n))^{1/2})$.
\end{proof}

Corollary \ref{cor1} follows directly by plugging the bounds on $m_i$ above.

For theorem \ref{thm:ext2}, we use another flexibility that we can reweight the rows of incidence matrix of the set system. In particular, if we scale the row $j$ by $w_j$ and still require the (weighted) vector discrepancy to be $\sqrt{t}$, then this implies a $\sqrt{t}/w_j$ discrepancy  on the original row. As long as this reweighting does not affect the $\ell_2$-norms by too much, the previous bounds and analysis still goes through, while allowing us to get 
tighter bounds on selected rows.  In particular we have the following.

\bigskip 

\noindent{\bf Theorem \ref{thm:ext2} (restated):} {\em
Suppose there is a reweighting of the rows $j$ with non-negative weights $w_j$ such that the $\ell_2$-norm of any column is at most $\beta$, then 
the algorithm gets a discrepancy of  $O(\beta(\log n \log (w_j s_j))^{1/2}/w_j)$ for each row $j$.
}

\begin{proof}(Sketch) 
The proof follows directly from Theorem \ref{thm:kom}. We work with the incidence matrix of reweighted system, and then translate the bounds back to the original system.
\end{proof}

%


Corollary \ref{cor2} follows from Theorem \ref{thm:ext2} as follows.

\begin{proof}(Corollary \ref{cor2})
Call a set a class-$i$ set for $i=1,\ldots,\log n$, if its size lies in the range $[t 2^{i-1},t2^{i})$.
We choose $w_i = \sqrt{i+\log t}$ for class $i$ sets, and $w_i=1$ for sets of size less than $t$. Fix some element $v$, and let $m_{\geq i}$ denote the number of sets in classes $i$ or more that contain $v$. Then, the squared $\ell_2$-norm of columns in the weighted incidence matrix is at most 
  \[ t + \sum_{i=1}^{\log n} w_i^2 (m_{\geq i} - m_{\geq i+1}) = t +  \sum_{i=1}^{\log n} (w_i^2-w_{i-1}^2) (m_{\geq i}) \leq 
	t + \sum_{i=1}^{\log n} m_{\geq i} = O(t) \]
	where we use that $m_{\geq i} \leq O(t/i^{1+\epsilon})$ for any element $v$.
		
Applying theorem \ref{thm:ext2}, we obtain that weighted discrepancy for a class-$i$ set is $O( (t (i+\log t) \log n)^{1/2})$ and thus the actual discrepancy is $O((t (i+\log t) \log n)^{1/2}/w_i) = O((t \log n)^{1/2})$. For sets of size at most $t$ the discrepancy is $O((t\log n\log t)^{1/2})$.
\end{proof}

	

We remark that there are lots of additional flexibility in the approach that we did not use in the extensions above. These could be helpful in some other special instances, and hopefully perhaps even for the general Beck-Fiala problem. For example, we can choose the weights $w_i$ adaptively over time for sets that become dangerous (i.e.~incur much more discrepancy than expected), as was done in \cite{B10} to obtain the bounds for Spencer's result \cite{Spencer85}, provided not too many sets become dangerous. Another observation is that the concentration of energy deficit is so strong for the sets of size $\Omega(t \log n)^{1/2}$, that we can essentially assume that the number of alive variables in a set decreases deterministically by a constant factor in each phase.
Once the residual size of a set falls below $O((t \log n)^{1/2})$, we can simply discard it from the set system (as is done in iterated rounding, as it can only incur $O( (t \log n)^{1/2})$ discrepancy henceforth). So what really matters at any time in the algorithm is the bound on the $\ell_2$-norm of any column of this reduced system, rather than the $\ell_2$-norm of columns in the original system. This reduction in the $\ell_2$-norm can potentially be used to tradeoff the high discrepancy   
incurred by some sets.

%% file: beck-fiala2.bbl
\newcommand{\etalchar}[1]{$^{#1}$}
\begin{thebibliography}{BCKL14}

\bibitem[Ban98]{B97}
Wojciech Banaszczyk.
\newblock Balancing vectors and gaussian measures of n-dimensional convex
  bodies.
\newblock {\em Random Structures \& Algorithms}, 12(4):351--360, 1998.

\bibitem[Ban10]{B10}
Nikhil Bansal.
\newblock Constructive algorithms for discrepancy minimization.
\newblock In {\em Foundations of Computer Science (FOCS)}, pages 3--10, 2010.

\bibitem[BCKL14]{BCKL14}
Nikhil Bansal, Moses Charikar, Ravishankar Krishnaswamy, and Shi Li.
\newblock Better algorithms and hardness for broadcast scheduling via a
  discrepancy approach.
\newblock In {\em {SODA}}, pages 55--71, 2014.

\bibitem[BF81]{BF81}
J{\'o}zsef Beck and Tibor Fiala.
\newblock Integer-making theorems.
\newblock {\em Discrete Applied Mathematics}, 3(1):1--8, 1981.

\bibitem[BN15]{BN15}
Nikhil Bansal and Viswanath Nagarajan.
\newblock Approximation-friendly discrepancy rounding.
\newblock {\em CoRR}, abs/1512.02254, 2015.

\bibitem[Buk13]{B13}
Boris Bukh.
\newblock An improvement of the beck-fiala theorem.
\newblock {\em CoRR}, abs/1306.6081, 2013.

\bibitem[Cha00]{Chazelle}
Bernard Chazelle.
\newblock {\em The discrepancy method: randomness and complexity}.
\newblock Cambridge University Press, 2000.

\bibitem[CNN11]{CNN11}
Moses Charikar, Alantha Newman, and Aleksandar Nikolov.
\newblock Tight hardness results for minimizing discrepancy.
\newblock In {\em SODA}, pages 1607--1614, 2011.

\bibitem[CST{\etalchar{+}}14]{Panorama}
William Chen, Anand Srivastav, Giancarlo Travaglini, et~al.
\newblock {\em A Panorama of Discrepancy Theory}, volume 2107.
\newblock Springer, 2014.

\bibitem[EL15]{EL15}
Esther Ezra and Shachar Lovett.
\newblock On the beck-fiala conjecture for random set systems.
\newblock {\em Electronic Colloquium on Computational Complexity {(ECCC)}},
  22:190, 2015.

\bibitem[ES14]{ES14}
Ronen Eldan and Mohit Singh.
\newblock Efficient algorithms for discrepancy minimization in convex sets.
\newblock {\em CoRR}, abs/1409.2913, 2014.

\bibitem[Fre75]{Freedman}
David~A. Freedman.
\newblock On tail probabilities for martingales.
\newblock {\em Annals of Probability}, 3:100--118, 1975.

\bibitem[HSS14]{HSS14}
Nicholas J.~A. Harvey, Roy Schwartz, and Mohit Singh.
\newblock Discrepancy without partial colorings.
\newblock In {\em {APPROX/RANDOM} 2014}, pages 258--273, 2014.

\bibitem[LM12]{LM12}
Shachar Lovett and Raghu Meka.
\newblock Constructive discrepancy minimization by walking on the edges.
\newblock In {\em Foundations of Computer Science (FOCS)}, pages 61--67, 2012.

\bibitem[LSV86]{LSV}
L{\'a}szl{\'o} Lov{\'a}sz, Joel Spencer, and Katalin Vesztergombi.
\newblock Discrepancy of set-systems and matrices.
\newblock {\em European Journal of Combinatorics}, 7(2):151--160, 1986.

\bibitem[Mat09]{Mat99}
Jiri Matousek.
\newblock {\em Geometric discrepancy: An illustrated guide}.
\newblock Springer Science, 2009.

\bibitem[Nik13]{N13}
Aleksandar Nikolov.
\newblock The koml\'{o}s conjecture holds for vector colorings.
\newblock {\em arXiv preprint arXiv:1301.4039}, 2013.

\bibitem[NT15]{NT15}
Aleksandar Nikolov and Kunal Talwar.
\newblock Approximating hereditary discrepancy via small width ellipsoids.
\newblock In {\em Symposium on Discrete Algorithms, {SODA}}, pages 324--336,
  2015.

\bibitem[NTZ13]{NTZ13}
Aleksandar Nikolov, Kunal Talwar, and Li~Zhang.
\newblock The geometry of differential privacy: the sparse and approximate
  cases.
\newblock In {\em Symposium on Theory of Computing , STOC}, pages 351--360,
  2013.

\bibitem[Rot12]{R12}
Thomas Rothvoss.
\newblock The entropy rounding method in approximation algorithms.
\newblock In {\em Symposium on Discrete Algorithms (SODA)}, pages 356--372,
  2012.

\bibitem[Rot13]{R13}
Thomas Rothvoss.
\newblock Approximating bin packing within o(log {OPT} * log log {OPT)} bins.
\newblock In {\em {FOCS}}, pages 20--29, 2013.

\bibitem[Rot14]{Ro14}
Thomas Rothvoss.
\newblock Constructive discrepancy minimization for convex sets.
\newblock In {\em Foundations of Computer Science (FOCS)}, pages 140--145,
  2014.

\bibitem[Spe85]{Spencer85}
Joel Spencer.
\newblock Six standard deviations suffice.
\newblock {\em Transactions of the American Mathematical Society},
  289(2):679--706, 1985.

\bibitem[Sri97]{Srin97}
Aravind Srinivasan.
\newblock Improving the discrepancy bound for sparse matrices: Better
  approximations for sparse lattice approximation problems.
\newblock In {\em Symposium on Discrete Algorithms (SODA)}, pages 692--701,
  1997.

\end{thebibliography}
